\documentclass{llncs}

\usepackage{amsmath,amssymb,amsfonts,latexsym}
\usepackage{enumitem}
\usepackage{graphicx,wrapfig}
\usepackage{tikz}

\usetikzlibrary{arrows, automata, positioning, fit, calc, patterns}

\renewcommand{\P}{\bbbp}

\newcommand{\coNP}{\mathbf{coNP}}

\newcommand{\aout}{\mathtt{Out}}

\newcommand{\ft}{\preceq}

\newcommand{\R}{\bbbr}

\newcommand{\N}{\bbbn}

\newcommand{\set}[1]{\left\{ #1 \right\}}

\newcommand{\A}{\mathcal{A}}
\newcommand{\B}{\mathcal{B}}
\newcommand{\C}{\mathcal{C}}
\newcommand{\Convex}{\mathrm{Convex}}
\newcommand{\Conv}{\mathrm{Conv}}

\newcommand{\M}{\mathcal{M}}

\newcommand{\D}{\mathcal{D}}

\newcommand{\wrt}[1]{\mathop{\mathrm{d}#1}}

\title{Timed Comparisons of Semi-Markov Processes\thanks{This work was supported by the Alan Turing Institute under the EPSRC grant EP/N510129/1, as well as by the
Danish FTP project ASAP, the ERC Advanced Grant LASSO, and the Sino-Danish Basic Research Center IDEA4CPS.}}
\author{Mathias R. Pedersen\inst{1} \and Nathana\"{e}l Fijalkow\inst{2} \and Giorgio Bacci \inst{1} \and Kim G. Larsen\inst{1} \and Radu Mardare\inst{1}}
\institute{Department of Computer Science, Aalborg University, Denmark \and The Alan Turing Institute, London, United Kingdom}

\begin{document}

\maketitle

\begin{abstract}
  Semi-Markov processes are Markovian processes
  in which the firing time of the transitions is modelled by probabilistic distributions over positive reals
  interpreted as the probability of firing a transition at a certain moment in time.

  In this paper we consider the trace-based semantics of semi-Markov processes, 
  and investigate the question of how to compare two semi-Markov processes with respect to their time-dependent behaviour.
  To this end, we introduce the relation of being ``faster than'' between processes and study its algorithmic complexity.
  Through a connection to probabilistic automata we obtain hardness results showing in particular that this relation is undecidable.
  However, we present an additive approximation algorithm for a time-bounded variant of the faster-than problem 
  over semi-Markov processes with slow residence-time functions, and a $\coNP$ algorithm 
  for the exact faster-than problem over unambiguous semi-Markov processes.
\end{abstract}

\section{Introduction}
Semi-Markov processes are Markovian stochastic systems
that model the firing time of transitions as probabilistic distribution over positive reals;
thus, one can encode the probability of firing a certain transition within a certain time interval.
For example, continuous-time Markov processes are particular case of semi-Markov processes where the timing distributions are always exponential.

Semi-Markov processes have been used extensively to model real-time systems such as power plants~\cite{PST97} and power supply units~\cite{PTV04}.
For such real-time systems, non-functional requirements are becoming increasingly important.
Many of these requirements, such as response time and throughput, depend heavily on the timing behaviour of the system in question.
It is therefore natural to understand and be able to compare the timing behaviour of different systems. 
 
Moller and Tofts~\cite{MT91} proposed the notion of a \emph{faster-than} relation for systems with discrete-time in the context of process algebras.
Their goal was to be able to compare processes that are functionally behaviourally equivalent,
except that one process may execute actions faster than the other.
This line of study was continued by L{\"{u}}ttgen and Vogler~\cite{LV01}, who moreover considered
upper bounds on time, in order to allow for reasoning about worst-case timing behaviours.
For timed automata, Guha et al.~\cite{GNA12} introduced a bisimulation-like faster-than relation
and studied its compositional properties.
For continuous-time probabilistic systems, Baier et al.~\cite{BKHW05} considered a simulation relation
where the timing distribution on each state is required to stochastically dominate the other.
They introduced both a weak and a strong version of their simulation relation,
and gave a logical characterization of these in terms of the logic CSL.

In the literature, less attention has been drawn to trace-based notions of faster-than relations
although trace equivalence and inclusion are important concepts when considering linear-time properties such as liveness or safety~\cite{BK08}.
In this paper we propose a simple and intuitive notion of trace 
inclusion for semi-Markov processes, which we call \emph{faster-than} relation, that compares the relative speed of
processes with respect to the execution of arbitrary sequences of actions.

Differently from trace inclusion, 
our relation does not make a step-wise comparison of the timing delays for each individual action in a sequence, but over the overall execution time of the sequence. 
As an example, consider the semi-Markov process in Fig. \ref{fig:faster-than}.
The states $s$ and $s'$, although performing the same sequences of actions,
are not related by trace inclusion because the first two actions in any sequence
are individually executed at opposite order of speeds (here governed by exponential-time distributions).
Instead, according to our relation, $s$ \emph{is} faster-than $s'$ (but not vice versa) because it executes single-action sequences at a faster rate than $s'$,
and action sequences of length greater than one at the same speed ---this is due to the fact that the execution time of each action is governed by random variables
that are independent of each other and the sum of independent random variables is commutative.

\begin{figure}
  \centering
%
\includegraphics[scale=.3]{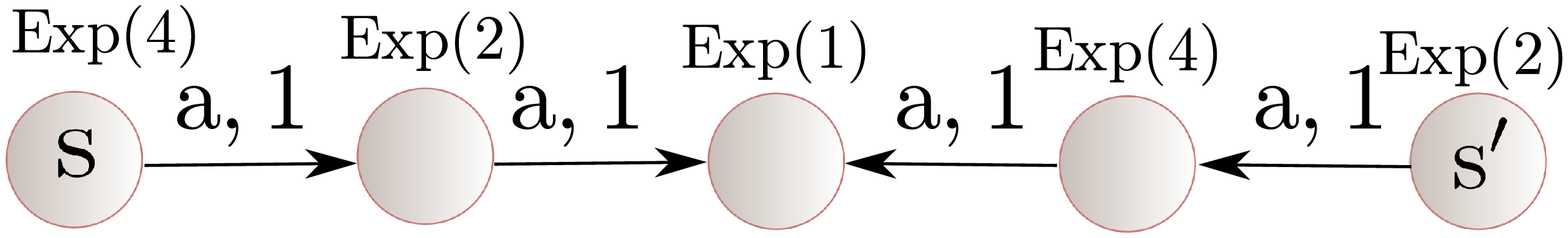}
  \caption{A semi-Markov process where $s$ is faster than $s'$. The states of the process are annotated with their timing distributions 
           and each action-labelled transition is decorated with its probability to be executed.}
  \label{fig:faster-than}
\end{figure}

%
%
%
 
In this paper we investigate the algorithmic complexity of various problems regarding the faster-than relation,
emphasising their connection with classical algorithmic problems over Rabin's probabilistic automata.
In particular, we prove that the faster-than problem over generic semi-Markov processes is undecidable and that it is 
Positivity-hard when restricted to processes with only one action label. The reduction from the Positivity
problem is important because it relates the faster-than problem to the Skolem problem,
an important problem in number theory, whose decidability status has been an open problem for 
at least 80 years~\cite{OW14,AAOW15}.

We show that undecidability for the faster-than problem can not be tackled even by approximation techniques:
via the same connection with probabilistic automata we are able to prove that the faster-than 
problem can not be approximated up to a multiplicative constant.
However, as a positive result, we show that a time-bounded variant of the faster-than problem,
which compares processes up to a given finite time bound, although still undecidable,
admits approximated solutions up to an \emph{additive} constant over semi-Markov processes with slow residence-time distributions.
These include the important cases of uniform and exponential distributions.

Finally, we present a $\coNP$ algorithm for solving the faster-than problem exactly over unambiguous semi-Markov processes,
where a process is unambiguous if every transition to a next state is unambiguously determined by the label that it outputs.

%


\section{Definitions}
For a finite set $S$ we let $\D(S)$ denote the set of (sub)distributions over $S$, i.e. functions $\delta : S \to [0,1]$
such that $\sum_{s \in S} \delta(s) \le 1$.
The subset of total distributions is $\D_{=1}(S)$.

We let $\bbbn$ denote the natural numbers and $\R_{\ge 0}$ denote the non-negative real numbers.
We equip $\R_{\ge 0}$ with the Borel $\sigma$-algebra $\B$,
so that $(\R_{\ge 0}, \B)$ is a measurable space.
Let $\D(\R_{\ge 0})$ denote the set of (sub)distributions over $(\R_{\ge 0}, \B)$,
i.e. measures $\mu : \B \to [0,1]$ such that $\mu(\R_{\ge 0}) \le 1$.
Throughout the paper we will write $\mu(t)$ for $\mu([0,t])$.

To avoid confusion we will refer to $\mu$ in $\D(\R_{\ge 0})$ as timing distributions,
and to $\delta$ in $\D(S)$ as distributions.

\begin{definition}[Semi-Markov process]
A \emph{semi-Markov process}, usually written $\M$, is given by:
  \begin{itemize}
    \item $S$ is a (finite) set of \emph{states},
    \item $\aout$ is a (finite) set of \emph{output labels},
    \item $\Delta : S \to \D(S \times \aout)$ is a \emph{transition function},
    \item $\rho : S \to \D(\R_{\ge 0})$ is a \emph{residence-time function}.
  \end{itemize}
\end{definition}

The operational behaviour of a semi-Markov process can be described as follows.
In a given state $s \in S$,
the process fires a transition within time $t$ with probability $\rho(s)(t)$,
leading to the state $s' \in S$ while outputting the label $a \in \aout$ with probability $\Delta(s)(s',a)$.

We aim at defining $\P_\M(s,w,t)$, 
the probability that from the state~$s$,
the output of the semi-Markov process $\M$ within time $t$ \emph{starts with} the word $w$.
It is important to note here that time is accumulated:
we sum together the time spent in all states along the way,
and ask that this total time is less than the specified bound $t$.
A full and formal definition of the probability can be done through the usual cylinder construction.
However, we will spare the reader this well-known construction and give seemingly ad-hoc definitions in this conference version.

In order to account for the accumulated time in the probability,
we need the notion of convolution. The convolution of two timing distributions $\mu$ and $\nu$ is $\mu * \nu$ defined by
\[
(\mu * \nu)(E) = \int_0^\infty \nu(E - x) \mu(\wrt{x})
\]
for any Borel set $E$.
Convolution is both associative and commutative.
Let $X$ and $Y$ be two independent random variables with timing distributions $\mu$ and $\nu$,
i.e. $\P(X \in E) = \mu(E)$ and $\P(Y \in E) = \nu(E)$,
then 
\[
\P(X + Y \in E) = (\mu * \nu)(E) \enspace.
\]


\begin{definition}[Probability]
  Consider a semi-Markov process $\M$. 
  We define the timing distribution $\P_\M(s,w)$ inductively:
  $\P_\M(s,\varepsilon) = \bbbone$ for the empty word $\varepsilon$,
  where $\bbbone$ is the function such that $\bbbone(t) = 1$ for all $t$ in $\R_{\ge 0}$,
  and for a word $w$ in $\aout^*$, a letter $a$ in $\aout$ and a state $s$,
  \[\P_\M(s,aw) = \sum_{s' \in S} \Delta(s)(s',a) \cdot \left(\rho(s) * \P_\M(s',w)\right) \enspace.\]
  We will then write $\P_\M(s,w,t)$ to mean $\P_\M(s,w)(t)$.
\end{definition}


\subsection*{Timed Comparisons}
We introduce the following relation which will be the focus of our paper.

\begin{definition}[Faster-than relation]
  Consider a semi-Markov process $\M$ and two states $s$ and $s'$.
  We say that $s$ is \emph{faster than} $s'$, denoted $s \ft s'$,
  if for all $w$, for all $t$,
  \[
  \P_\M(s,w,t) \ge \P_\M(s',w,t) \enspace.
  \]
\end{definition}

The algorithmic problem we consider in this paper is the \emph{faster-than problem}: 
given a semi-Markov process and two states $s$ and $s'$, determine whether $s \ft s'$.

\subsection*{Algorithmic Considerations}

The definition we use for semi-Markov processes is very general, because we allow for any residence-time function.
The aim of the paper is to give generic algorithmic results which apply to \emph{effective} classes of timing distributions,
a notion we define now.
Recall that a residence-time function associates with each state a timing distribution.
We first give some examples of classical timing distributions.
\begin{itemize}
	\item The prime example is exponential distributions, defined by the timing distribution
	$\mu(t) = 1 - e^{-\lambda t}$ for some parameter $\lambda > 0$ usually called the rate.
	\item Another interesting example is piecewise polynomial distributions. Consider finitely many polynomials $P_1,\ldots,P_n$
	and a finite set of pairwise disjoint intervals $I_1 \cup I_2 \cup \cdots \cup I_n$ covering $[0,\infty)$ such that 
	for every $k$, $P_k$ is non-negative over $I_k$ and $\sum_k \int_{I_k} P_k = 1$.
	This induces the timing distribution
	\[
	\mu(t) = \sum_k \int_{I_k \cap [0,t]} P_k(t) \enspace.
	\]
	\item A special case of the previous example is given by piecewise affine distributions, where the polynomials are affine functions.
	\item Another important special case of piecewise polynomial distributions are the uniform distributions with parameters $0 \leq a < b$ 
	defining the timing distribution
	\[
	\mu(t) = 
	\begin{cases}
	1 & \mbox{if } t < a, \\
	\frac{t-a}{b-a} & \mbox{if } t \in [a,b) \\ 
	0 & \mbox{if } x \ge b \enspace.
	\end{cases}
	\]
	\item The simplest example is given by Dirac distributions defined for the parameter $a$ by 
	$\mu(E) = 1$ if $a$ is in $E$, and $0$ otherwise.
\end{itemize}

The following definition captures these examples, and more.
For a class $\C$ of timing distributions, we let $\Convex(\C)$ be the smallest class of timing distributions containing $\C$ 
and closed under convex combinations, 
and similarly $\Conv(\C)$ adding closure under convolutions.

\begin{lemma}
  \label{lem:convC}
  Let $\C$ be a class of timing distributions.
  Consider a semi-Markov process $\M$ whose residence-time function uses timing distributions from $\C$,
  a state $s$ and a word $w$, then $\P_\M(s,w) \in \Conv(\C)$.
\end{lemma}

Lemma~\ref{lem:convC} is established by a straightforward induction on the word $w$ using the definition of $\P_\M(s,w)$.

In the rest of the paper we will consider only distributions that are suitable for algorithmic manipulation. Clearly, we must be able to give them as input to a computational device, so we assume they can be described by finitely many rational parameters. Moreover, we require that testing inequalities between them is decidable, since this is essential for determining the 
faster-than relation. The next definition formalises this intuition.

\begin{definition}[Effective timing distributions]
  A class $\C$ of timing distributions is \emph{effective} if, for any 
  $\varepsilon \ge 0$, $b \in \R_{\ge 0} \cup \set{\infty}$, and $\mu_1,\mu_2 \in \Conv(\C)$, 
  it is decidable whether $\mu_1(t) \ge \mu_2(t) - \varepsilon$, for all $t \le b$.
\end{definition}

\begin{proposition}\label{thm:effective_classes}
  The following classes of timing distributions are effective:
  \begin{itemize}
	  \item exponential distributions,
	  \item piecewise polynomial distributions,
	  \item piecewise affine distributions,
	  \item uniform distributions,
	  \item Dirac distributions.
  \end{itemize}
\end{proposition}

We do not provide in the conference version a full proof of Proposition~\ref{thm:effective_classes},
as it is mostly folklore but rather tedious.
In particular, for exponential and piecewise polynomial distributions one relies on decidability results 
for the existential theory of the reals~\cite{Canny88,Tarski51}, implying that the most demanding operations above can be performed in polynomial space.

An effective class $\C$ of timing distributions induces the set of semi-Markov processes whose residence-time functions
use timing distributions from $\C$.
Furthermore, a given semi-Markov process has only finitely many states,
and hence can only use finitely many timing distributions.
For our decidability results we will therefore focus on finite classes of timing distributions.
This paper gives algorithmic results for generic effective classes of timing distributions.
In our complexity analyses, we will always assume that the operations on the timing distributions have a unit cost.


\section{Hardness Results}
\label{sec:hardness}
We start the technical part of this article by hardness results inherited from Markov processes.
A Markov process is a semi-Markov process without the residence-time function,
and for a Markov process $\M = (S, \aout, \Delta)$,
we define the probability $\P_\M(s,aw) = \sum_{s' \in S} \Delta(s)(s',a) \cdot \P_\M(s')(w)$
and $\P_\M(s,\varepsilon) = 1$ for the empty word.
The faster-than relation for Markov processes is then $s \preceq s'$ if for all $w$ we have $\P_\M(s,w) \geq \P_\M(s',w)$.

We show that the faster-than problem for Markov processes, and hence also for semi-Markov processes, is undecidable in general,
can not be multiplicatively approximated, and relates to an open problem in number theory even in a restricted case.
These limitations shape and motivate our positive results,
which will be the topic of the remaining sections.

We first explain how hardness results for Markov processes directly imply hardness results for semi-Markov processes.
The following lemma formalises the two ways semi-Markov processes subsume Markov processes.

\begin{lemma}\label{lem:Markov_process_reduction}
  Consider a semi-Markov process $\M = (S,\aout,\Delta,\rho)$ and its induced Markov process $\M' = (S,\aout,\Delta)$.
  \begin{itemize}
	  \item If $\rho$ is constant, i.e. for all $s,s'$ we have $\rho(s) = \rho(s')$,
	  then for all $w$, for all $t$, we have $\P_\M(s,w,t) = \P_{\M'}(s,w) \cdot 
	  (\underbrace{\rho(s) * \cdots * \rho(s)}_{|w| \text{ times}})(t)$.
	  \item If for all $s$, $\rho(s)$ is the Dirac distribution for $0$, 
	  then for all $w$, for all $t$, we have $\P_\M(s,w,t) = \P_{\M'}(s,w)$.
  \end{itemize}
  In particular in both cases, the following holds:
  for $s,s'$ two states, we have $s \ft s'$ in $\M$ if, and only if, $s \ft s'$ in $\M'$.
\end{lemma}

We will use Lemma~\ref{lem:Markov_process_reduction} to draw corollaries about semi-Markov processes from hardness results of Markov processes.

The hardness results of this section will be based on a connection to probabilistic automata.
A probabilistic automaton is given by 
\[
\A = (Q,A,q_0,\Delta : Q \times A \to \D_{=1}(Q),F) \enspace,
\]
where $Q$ is the state space, $A$ is the alphabet, $q_0$ is an initial state, $\Delta$ is the transition function,
and $F$ is a set of final or accepting states.
Any probabilistic automaton $\A$ induces the probability $\P_\A(w)$ that a run over $w \in A^*$ is accepting,
i.e. starts in $q_0$ and ends in $F$.
The key property of probabilistic automata that we will exploit is the undecidability of the universality problem, 
which was proved in~\cite{Paz71}, see also~\cite{GO10}.
The universality problem is as follows: given a probabilistic automaton $\A$, determine whether for all words $w$ in $A^+$ we have 
$\P_\A(w) \ge \frac{1}{2}$.

We describe a construction which given a probabilistic automaton $\A$, constructs the \emph{derived} Markov process $\M(\A)$.
The set of states of $\M(\A)$ is $Q \times \set{\ell,r} \cup \set{\top}$, where $\top$ is a new state.
Let $s = (q_0,\ell)$ and $s' = (q_0,r)$, where $q_0$ is the initial state of $\A$.
The set of output labels is $A$, and the transition function $\Delta'$ is defined as follows:
\begin{align*}
  \Delta'(p,\ell)((q,\ell),a) &= \frac{1}{2|A|} \Delta(p,a)(q)  &\quad \quad
  \Delta'(p,\ell)(\top,a)     &= \frac{1}{2|A|} \mbox{ if } p \in F \\
  \Delta'(p,r)((q,r),a)       &= \frac{1}{2|A|} \Delta(p,a)(q)  &\quad \quad
  \Delta'(p,r)(\top,a)        &= \frac{1}{4|A|} \enspace.
\end{align*}
We can then verify the following equalities:
  \[
  \P_{\M(\A)}(s,wa)	= \frac{1}{(2 |A|)^{|w| + 1}} \left(1 + \P_\A(w) \right) 
  \]
  and
  \[
  \P_{\M(\A)}(s',wa)= \frac{1}{(2 |A|)^{|w| + 1}} \left(1 + \frac{1}{2} \right) \enspace.
  \]

\begin{theorem}\label{thm:fasterthan_undecidability}
  The faster-than problem is undecidable for Markov processes.
\end{theorem}

\begin{proof}
  Given a probabilistic automaton $\A$, we construct
  the derived Markov process $\M(\A)$.
  Thanks to the equalities above, $\A$ is universal if, and only if, $s \ft s'$.
  \qed

\end{proof}


We discuss three approaches to recover decidability.

A first approach is to look for \emph{structural restrictions} on the underlying graph.
However, the undecidability result above for probabilistic automata is quite robust in this aspect,
as it already applies when the underlying graph is acyclic, meaning that the only loops are self-loops.
In spite of this, we present in Sect.~\ref{sec:unambiguous} an algorithm to solve the faster-than problem
for \emph{unambiguous} semi-Markov processes.

A second approach is to restrict the \emph{observations}.
The undecidability result above holds already when there are two different output letters,
hence a natural question is to look at what happens when we only have one output letter.
Interestingly, specialising the construction above yields a reduction from the Positivity problem.
This problem appears in various contexts, prominently in number theory,
and its decidability status has been an open problem for at least 30 years~\cite{OW14}.
Formally, the Positivity problem reads: given a linear recurrence sequence, are all terms of the sequence non-negative?
It has been shown that the universality problem for probabilistic automata with one letter alphabet is equivalent 
to the Positivity problem~\cite{AAOW15}.
Thus, using again the derived Markov process $\M(\A)$ for a probabilistic automaton $\A$ with only one label, we obtain the following result.

\begin{theorem}
  The faster-than problem is Positivity-hard over Markov processes with one output label.
\end{theorem}


A third approach is \emph{approximations}.
However, we can exploit further the connection we made with probabilistic automata, 
obtaining an impossibility result for \emph{multiplicative approximation}.
We rely on the following celebrated theorem for probabilistic automata due to Condon and Lipton \cite{CL89}.
The following formulation of their theorem is described in detail in \cite{Fijalkow17}.

\begin{theorem}[\cite{CL89}]
\label{thm:inapproximability_automata}
Let $0 < \alpha < \beta < 1$ be two constants.
There is no algorithm which, given a probabilistic automaton $\A$,
  \begin{itemize}
    \item if for all $w$ we have $\P_\A(w) \geq \beta$, returns YES,
    \item if there exists $w$ such that $\P_\A(w) \leq \alpha$, returns NO.
  \end{itemize}
\end{theorem}

\begin{theorem}\label{thm:impossibility_approximation}
  Let $0 < \varepsilon < \frac{1}{3}$ be a constant.
  There is no algorithm which, given a Markov process $\M$ and two states $s,s'$,
  \begin{itemize}
    \item if for all $w$ we have $\P_\M(s,w) \geq \P_\M(s',w)$, returns YES,
    \item if there exists $w$ such that $\P_\M(s,w) \leq \P_\M(s',w) \cdot (1 - \varepsilon)$, returns NO.
  \end{itemize}
\end{theorem}
\begin{proof}
  Assume towards a contradiction that there exists an algorithm as described in the theorem.
  We then construct an algorithm satisfying the specifications of Theorem~\ref{thm:inapproximability_automata}.

  Let $\alpha = \frac{1}{2} - \frac{3\varepsilon}{2}$ and $\beta = \frac{1}{2}$,
  and let $\A$ be a probabilistic automaton.
  We now run the algorithm on the derived Markov process $\M(\A)$.
  \begin{itemize}
      \item If for all $w$ we have $\P_{\M(\A)}(s,w) \geq \P_{\M(\A)}(s',w)$, then the algorithm returns YES. 
      Indeed, this is equivalent to $\P_\A(w) \ge \beta$.
      \item If there exists $w$ such that $\P_{\M(\A)}(s,w) \leq \P_{\M(\A)}(s',w) \cdot (1 - \varepsilon)$, then the algorithm returns NO.
      Indeed, this is equivalent to $\P_\A(w) \le \alpha$.
  \end{itemize}
  Hence we constructed an algorithm satisfying the specifications of Theorem~\ref{thm:inapproximability_automata},
  a contradiction.
  \qed
\end{proof}


From these hardness results for Markov processes together with Lemma \ref{lem:Markov_process_reduction},
we get the following hardness results for semi-Markov processes.

\begin{corollary}
  The following holds for semi-Markov processes for any class of timing distributions.
  \begin{itemize}
    \item The faster-than problem is undecidable.
    \item The faster-than problem with only one output label is Positivity-hard.
    \item The faster-than problem can not be multiplicatively approximated.
  \end{itemize}
\end{corollary}


\section{Time-Bounded Additive Approximation}
\label{sec:approximation}
Instead of considering multiplicative approximation,
we can also consider additive approximation,
meaning that we want to decide whether for all $w$ and $t$ we have $\P_\M(s,w,t) \geq \P_\M(s',w,t) - \varepsilon$
for some constant $\varepsilon > 0$.
In this section, we present an algorithm to solve the problem of approximating additively the faster-than relation with two assumptions:
\begin{itemize}
	\item \emph{time-bounded}: we only look at the behaviours up to a given bound $b$ in $\bbbr_{\ge 0}$,
	\item \emph{slow residence-time functions}: each transition takes \emph{some} time to fire.
\end{itemize}
As we will show, the combination of these two assumptions imply that the relevant words have bounded length.
This is in contrast to the impossibility of approximating the faster-than relation multiplicatively
that we showed in Sect. \ref{sec:hardness}.
More precisely, we consider the \emph{time-bounded} variant of the faster-than problem:
given a time bound $b$ in $\R_{\ge 0}$, a semi-Markov process, and two states $s$ and $s'$,
determine whether for all $t \leq b$ and $w$ it holds that $\P_\M(s,w,t) \geq \P_\M(s',w,t)$.

We first observe that this restriction of the faster-than problem
does not make any of the problems in Sect. \ref{sec:hardness} easier for semi-Markov processes. 
Indeed, if the residence-time functions are all Dirac distributions on $0$, then all transitions are fired instantaneously, and
the time-bounded restriction is immaterial.
Thus we focus on distributions that do not fire instantaneously,
as made precise by the following definition.


\begin{definition}[Slow distributions]
  We say that a class $\C$ of timing distributions is \emph{slow} if 
  for all finite subset $\C_0$ of $\C$, 
  there exists a computable function $\varepsilon \colon \N \times \R_{\ge 0} \to [0,1]$ such that 
  for all $n$, $t$, and $\mu_1, \dots,\mu_n \in \Convex(\C_0)$ we have 
  $(\mu_1 * \dots *\mu_n)(t) \leq \varepsilon(n,t)$
  and $\lim_{n \to \infty} \varepsilon(n,t) = 0$.
\end{definition}

Given a slow and effective class $\C$ of timing distributions,
we can do additive approximation of the time-bounded faster-than problem in the following way.
We introduce the following notation.
Fix a semi-Markov process $\M$.
Let $\C_\M = \Convex(\set{\rho(s) \mid s \in S})$, and $n$ in $\N$.
We define the timing distribution $F_{\M,n}$ by $F_{\M,n}(t) = 1$ if $n = 0$ and otherwise
\[
  F_{\M,n}(t) = \sup \set{(\mu_1 * \cdots * \mu_n)(t) \mid \mu_1,\ldots,\mu_n \in \C_\M} \enspace.
\]

\begin{lemma}\label{lem:upperbound_f}
  For all $s$ and all $w$, we have $\P_\M(s,w) \le F_{\M,|w|}$.
\end{lemma}

\begin{proof}
  We proceed by induction on the length of $w$.
  It is clear for $|w| = 0$.
  \begin{align*}
    \P_\M(s,aw) 
    &= \sum_{s' \in S} \Delta(s)(s',a) \cdot \rho(s) * \P_\M(s',w)  \\
    &\le \underbrace{\sum_{s' \in S} \Delta(s)(s',a) \cdot \rho(s)}_{\in \C_\M} * F_{\M,|w|} \\
    &\le F_{\M,|w|+1} \enspace.
  \end{align*}
  This concludes.
  \qed
\end{proof}

\begin{theorem}
  \label{thm:approx_faster}
  There exists an additive approximation algorithm for the time-bounded faster-than problem over semi-Markov processes
  for all slow and effective classes of timing distributions.

  In other words, for a constant $\varepsilon > 0$, there exists an algorithm which, given a semi-Markov process $\M$,
  two states $s,s'$, and a bound $b$ in $\R_{\geq 0}$,
  determines whether
  \[
  \forall w, \forall t \le b,\ \P_\M(s,w,t) \ge \P_\M(s',w,t) - \varepsilon \enspace.
  \]
\end{theorem}

\begin{proof}
Let $\C_\M = \Convex(\set{\rho(s) \mid s \in S})$, since $S$ is finite there exists 
a computable function $\varepsilon \colon \N \times \R_{\ge 0} \to [0,1]$ such that 
for all $n$, $t$, and $\mu_1, \dots,\mu_n \in \C_\M$ we have 
$(\mu_1 * \dots *\mu_n)(t) \leq \varepsilon(n,t)$
and $\lim_{n \to \infty} \varepsilon(n,t) = 0$.
Given $\varepsilon > 0$, there exists $N$ such that $\varepsilon(N,b) < \varepsilon$.
Let $n \ge N$. By assumption
\[
(\mu_1 * \dots * \mu_n)(b) \leq \varepsilon(n,b) \le \varepsilon(N,b) < \varepsilon
\] 
for all $\mu_1, \dots, \mu_n \in \C_\M$.
Taking the supremum over $\mu_1, \dots, \mu_n$, we then get $F_{\M, n}(b) < \varepsilon$,
and by Lemma \ref{lem:upperbound_f}, this means that for all $w$ of length at least $N$, we have $\P_\M(s',w,b) < \varepsilon$.
Hence it holds trivially that for all $t \leq b$ and $w$ of length at least $N$, we have $\P_\M(s,w,t) \geq \P_\M(s',w,t) - \varepsilon$.

Thus the algorithm checks whether for all words of length less than $N$, for all $t \le b$, we have 
$\P_\M(s,w,t) \geq \P_\M(s',w,t) - \varepsilon$, which is decidable thanks to the effectiveness of $\C$.
  \qed
\end{proof}

Next we show that there are interesting classes of timing distributions that are indeed slow.
For this we introduce a class of timing distributions that are not just slow,
but furthermore are guaranteed to converge to zero rapidly.
We say that a timing distribution $\mu$ is \emph{very slow} if there exists a computable function $\varepsilon : \R_{\ge 0} \to [0,1]$ 
such that $\lim_{t \to 0} \frac{\varepsilon(t)}{t} = 0$
and for all $t$, we have $\mu(t) \le \varepsilon(t)$.
In order to show that very slow timing distributions are slow,
we need the following lemma.


\begin{lemma}\label{lem:bound_sum}
  Let $\mu_1,\ldots,\mu_n$ be timing distributions.
  Then 
  \[
  (\mu_1 * \mu_2 * \cdots * \mu_n)(t) \leq \sum_{i = 1}^n \mu_i \left(\frac{t}{n}\right) \enspace.
  \]
\end{lemma}

\begin{proof}
  We proceed by induction on $n$.
  The case of $n = 1$ is trivial.
  Recall that for any non-negative function $f$ and measure $\mu$ we have
  \begin{equation}\label{eq:ineq}
    \int_E f(x) \mu(\wrt{x}) \leq \mu(E) \cdot (\sup_E f(x)) \enspace.
  \end{equation}
  Let $\mu = \mu_1 * \dots * \mu_n$.
  \begin{align*}
		     & (\mu_1 * \cdots * \mu_{n+1})(t) \\
             &= \int_0^t \mu(t-x) \mu_{n+1}(\wrt{x}) \\
             &= \int_0^{\frac{nt}{n+1}} \mu(t-x) \mu_{n+1}(\wrt{x}) 
             + \int_{\frac{nt}{n+1}}^t \mu(t-x) \mu_{n+1}(\wrt{x}) \\
             &= \int_0^{\frac{nt}{n+1}} \mu(t-x) \mu_{n+1}(\wrt{x}) 
             + \int_0^{\frac{t}{n+1}} \mu \left(\frac{t}{n+1} - u \right) \mu_{n+1}(\wrt{u}) \\
             &\leq \mu \left(\frac{nt}{n+1}\right) 
             + \mu_{n+1}\left(t-\frac{nt}{n+1}\right) \\
             &\leq \sum_{i=1}^n \mu_i\left(\frac{n}{n+1} \frac{t}{n}\right) 
             + \mu_{n+1} \left(\frac{t}{n+1}\right) 
             = \sum_{i=1}^{n+1} \mu_i\left(\frac{t}{n+1}\right) \enspace.
  \end{align*}
  The third equality is the change of variable $u = x - \frac{nt}{n+1}$.
  The first inequality uses for each summand the inequality~\eqref{eq:ineq}.
  The second inequality is by induction hypothesis.
  \qed
\end{proof}







We can now prove the following theorem.

\begin{theorem}
  The following classes of timing distributions are slow:
  \begin{itemize}
    \item very slow distributions,
    \item uniform distributions, and
    \item exponential distributions.
  \end{itemize}
\end{theorem}
\begin{proof}
  Let $\C$ be a class of very slow timing distributions, and $\C_0 = \set{\mu_1,\ldots,\mu_n}$ a finite subset of~$\C$.
  Since every timing distribution in $\C$ is very slow,
  for every $i \in \set{1,\ldots,n}$ there exists a function $\varepsilon_i$ such that $\mu_i(t) \leq \varepsilon_i(t)$ for all $t$.
  Let $\varepsilon(n,t) = n \cdot \max \set{\varepsilon_i\left(\frac{t}{n} \right) \mid i \in \set{1,\ldots,n}}$.
  Note that $\lim_{n \to \infty} \varepsilon(n,t) = 0$.
  Let $\nu_1,\ldots,\nu_n$ in $\Convex(\C_0)$, we have $(\nu_1 * \cdots * \nu_n)(t) \le \sum_{i = 1}^n \nu_i \left(\frac{t}{n} \right)$
  thanks to Lemma~\ref{lem:bound_sum}.
  This implies that $(\nu_1 * \cdots * \nu_n)(t) \le \varepsilon(n,t)$, 
  which concludes.

  For exponential distributions, we proceed as follows.
  Let $\C_0$ be a finite class of exponential distributions.
  Let $\lambda > 0$ be the rate of the slowest exponential distributions appearing in $\C_0$,
  and let $\mu(t) = 1 - e^{-\lambda t}$.
  Then for any $\mu_1, \dots, \mu_n$ in $\Convex(\C_0)$ we have
  \[(\mu_1 * \dots * \mu_n)(t) \leq (\underbrace{\mu * \dots * \mu}_{n \text{ times}})(t).\]
  The distribution $\mu * \dots * \mu$ is called the Gamma (or more precisely, Erlang) distribution,
  and there is a computable closed form for it.
  In particular, if we let 
  \[\varepsilon(n,t) = (\underbrace{\mu * \dots * \mu}_{n \text{ times}})(t),\]
  we have $\lim_{n \to \infty} \varepsilon(n,b) = 0$,
  so exponential distributions are slow.

  Uniform distributions can be handled using a similar way as for exponential distributions.
  Let $\C_0$ be a finite class of uniform distributions with parameters $a_i$ and $b_i$ for $i \in \set{1,\ldots,n}$.
  Let $a$ be the smallest $a_i$ and $b$ the smallest $b_i$,
  and let $\mu$ be the uniform distribution with parameters $a$ and $b$.
  Then it follows that
  \[(\mu_1 * \dots * \mu_n)(t) \leq (\underbrace{\mu * \dots * \mu}_{n \text{ times}})(t) = \varepsilon(n,t).\]
  Then $(\mu * \dots * \mu)$ also has a nice closed form \cite{KvC10}
  and $\lim_{n \to \infty} \varepsilon(n,b) = 0$.
  \qed
\end{proof}




\section{Unambiguous Semi-Markov Processes}
\label{sec:unambiguous}
In order to regain decidability of the faster-than relation,
we can look at structurally simpler special cases of semi-Markov processes.
Here we will focus on semi-Markov processes such that each output word induces
at most one trace of states.
More precisely, we will say that a semi-Markov process is \emph{unambiguous} if 
for every $s$ in $S$ and $a$ in $\aout$, there exists at most one $s'$ in $S$ such that $\Delta(s)(s',a) \neq 0$.
A related notion of bounded ambiguity has been utilised to obtain decidability results in the context of probabilistic automata \cite{FRW17}.
We introduce the following notation for unambiguous semi-Markov processes:
$T(s,w)$ is the state reached after emitting $w$ from $s$.

\begin{example}
  Figure \ref{fig:unambiguous} gives an example of an unambiguous semi-Markov process.
  For each of the three states, there is at most one state that can be reached by a given output label.
  However, there need not be a transition for each output label from every state.
  In this example, the state $s_2$ has no $b$-transition, so for instance $T(s_1,ab) = s_2$, but $T(s_1,abb)$ is undefined.
\begin{figure}
  \centering
%
\includegraphics[scale=.3]{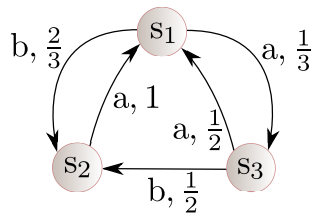}
  \caption{An example of an unambiguous semi-Markov process.}
  \label{fig:unambiguous}
\end{figure}
\end{example}

\begin{theorem}\label{thm:unambiguous}
  The faster-than problem is decidable in $\coNP$ over unambiguous semi-Markov processes for all effective classes of timing distributions.
\end{theorem}

Theorem~\ref{thm:unambiguous} follows from the next proposition.

\begin{proposition}\label{prop:unambiguous}
  Consider an unambiguous semi-Markov process $\M$ and two states $s,s'$.
  Let $L(s,s')$ be the set of loops reachable from $(s,s')$:
  \[
  \set{(p,p',v) \in S^2 \times \aout^{\le S^2} 
  \left|
  \ \exists w \in \aout^{\le S^2},\ 
  \begin{array}{c}
  T(s,w) = p,\ T(s',w) = p',\\
  T(p,v) = p,\ T(p',v) = p'
  \end{array}
  \right.} \enspace.
  \]

  We have $s \ft s'$ if, and only if
  \begin{itemize}
	  \item for all $w$ in $\aout^{\le S^2}$, we have $\P_\M(s,w) \ge \P_\M(s',w)$, and
	  \item for all $(p,p',v)$ in $L(s,s')$, we have $\P_\M(p, v) \ge \P_\M(p',v)$.
  \end{itemize} 
\end{proposition}

Before going into the proof, we explain how to use Proposition~\ref{prop:unambiguous}
to construct an algorithm solving the faster-than problem over unambiguous semi-Markov processes.
\begin{enumerate}
	\item The first step is to compute $L(s,s')$, which can be done in polynomial time using a simple graph analysis,
	\item The second step is to check the two properties, which both can be reduced 
	to exponentially many queries of the form: $\mu_1 \ge \mu_2$ for $\mu_1,\mu_2$ in $\Conv(\C)$.
\end{enumerate}
To obtain a $\coNP$ algorithm, in the second step we guess which of the two properties is not satisfied and a witness of polynomial length, 
which is either a word of quadratic length for the first property,
or two states and a word of quadratic length for the second property.

We split the proof of Proposition~\ref{prop:unambiguous} into two lemmas, each proving one direction of the proposition.
The following lemma gives the first direction.

\begin{lemma}
  \label{lem:hard_direction_unambiguous}
  If $s \ft s'$, then, for all $(p,p',v) \in L(s,s')$, $\P_\M(p, v) \ge \P_\M(p', v)$.
\end{lemma}

\begin{proof}
  Assume that $s$ is faster than $s'$ and let $(p,p')$ be in $L(s,s')$.
  There exist $w,v$ in $\aout^*$ such that 
  $T(s,w) = p,\ T(s',w) = p',\ T(p,v) = p,\ T(p',v) = p'$.
  Let $n$ in $\N$.
  Since $s$ is faster than $s'$, we have $\P_\M(s,w v^n) \ge \P_\M(s',w v^n)$.
  We have 
  \begin{align*}
    \P_\M(s, w v^n) &= \P_\M(s,w) * \underbrace{\P_\M(p,v) * \dots * \P_\M(p,v)}_{n \text{ times}} \\
    \P_\M(s',w v^n) &= \P_\M(s',w) * \underbrace{\P_\M(p',v) * \dots * \P_\M(p',v)}_{n \text{ times}} \enspace.
  \end{align*}
  Let $X_{s,w}$ be the random variable measuring the time elapsed from $s$ emitting $w$.
  Similarly, we define $X_{p,v}, Y_{s',w}$ and $Y_{p',v}$.
  We have: for all $n$ in $\N$, for all $t$,
  \[\P_\M(X_{s,w} + n X_{p,v} \leq t) \geq \P_\M(Y_{s',w} + n Y_{p',v} \leq t) \enspace,\]
  Dividing both sides by $n$ yields
  \[\P_\M \left(\frac{X_{s,w}}{n} + X_{p,v} \leq \frac{t}{n} \right) \geq \P_\M \left(\frac{Y_{s',w}}{n} + Y_{p',v} \leq \frac{t}{n} \right) \enspace.\]
  We make the change of variables $x = \frac{t}{n}$: for all $n$ in $\N$, for all $x$ we have
  \[\P_\M \left(\frac{X_{s,w}}{n} + X_{p,v} \leq x\right) \geq \P_\M \left(\frac{Y_{s',w}}{n} + Y_{p',v} \leq x\right) \enspace.\]
  Letting $n \rightarrow \infty$, we then obtain, for all $x$
  \[\P_\M(X_{p,v} \leq x) \geq \P_\M(Y_{p',v} \leq x) \enspace,\]
  which is equivalent to $\P_\M(p,v) \ge \P_\M(p',v)$.
  \qed
\end{proof}

The following lemma gives the converse implication of Proposition~\ref{prop:unambiguous}.

\begin{lemma}\label{lem:easy_direction_unambiguous}
  Assume that
  \begin{itemize}
	  \item for all $w$ in $\aout^{\le S^2}$, we have $\P_\M(s,w) \ge \P_\M(s',w)$, and
	  \item for all $(p,p',v)$ in $L(s,s')$, we have $\P_\M(p, v) \ge \P_\M(p', v)$.
  \end{itemize} 
  Then $s \ft s'$.
\end{lemma}

\begin{proof}
  We prove that for all $w$, we have $\P_\M(s,w) \ge \P_\M(s',w)$ by induction on the length of $w$.

  For $w$ of length at most $S^2$, this is ensured by the first assumption.
  Let $w$ be a word longer than $S^2$.
  There exist two states $p,p'$ such that $p$ is reached by $s$ and $p'$ by $s'$
  after emitting $i$ letters of $w$ and again after emitting $j$ letters of $w$,
  with $j$ at most $S^2$.
  Let $w = w_1\ v\ w_2$ where $v$ starts at position $i$ and ends at position $j$. 
  By construction $(p,p',v)$ is in $L(s,s')$.
  We have
  \begin{align*}
  \P_\M(s,w) &= \P_\M(s,w_1) * \P_\M(p,v) * \P_\M(p,w_2) \\
		         &= \P_\M(s,w_1) * \P_\M(p,w_2) * \P_\M(p,v) \\
		         &= \P_\M(s,w_1 w_2) * \P_\M(p,v) \\
		         &\ge \P_\M(s',w_1 w_2) * \P_\M(p',v) \\
		         &= \P_\M(s',w_1) * \P_\M(p',w_2) * \P_\M(p',v) \\
		         &= \P_\M(s',w_1) * \P_\M(p',v) * \P_\M(p',w_2) \\
             &= \P_\M(s',w) \enspace.
  \end{align*}
  The equalities use the associativity and commutativity of the convolution.
  The inequality $\P_\M(s,w_1 w_2) \ge \P_\M(s',w_1 w_2)$ holds by induction hypothesis, 
  because $w_1 w_2$ is shorter than $w$.
  The inequality $\P_\M(p,v) \ge \P_\M(p',v)$ holds thanks to the second assumption.
  \qed
\end{proof}


\section{Conclusion and Open Problems}
We studied the model of semi-Markov processes where the timing behaviour can be described by arbitrary timing distributions.
We have introduced a trace-based relation called the faster-than relation which asks that for any prefix and any time bound,
the probability of outputting a word with that prefix within the time bound is higher in the faster process than in the slower process.
We have shown through a connection to probabilistic automata that the faster-than relation is highly undecidable. 
It is undecidable in general, and remains Positivity-hard even for one output label.
Furthermore, approximating the faster-than relation up to a multiplicative constant is impossible.

However, we constructed algorithms for special cases of the faster-than problem.
We have shown that if one considers approximating up to an additive constant rather than a multiplicative constant,
and if one gives a bound on the time up to which one is interested in comparing the two processes,
then approximation can be done for timing distributions in which we are sure to spend some amount of time to take a transition.
In addition, we have shown that the faster-than relation is decidable and in $\coNP$ for unambiguous processes,
in which there is a unique successor state for every output label.

In this paper, we have focused on the generative model,
where the labels are treated as outputs.
An alternative viewpoint is the reactive model,
where the labels are instead treated as inputs~\cite{GSS95}.
While all the undecidability and hardness results we have shown
can also easily be shown to hold for the reactive case,
the same is not true for the algorithms we have constructed.
It is non-trivial to extend these algorithms to the reactive case,
and the main obstacle in doing so is that for reactive systems,
one has to also handle schedulers, often uncountably many.
It is therefore still an open question whether our decidability results carry over to reactive systems.


\bibliography{bib}

\end{document}